\documentclass[final,leqno,onefignum,onetabnum]{siamltex1213}
\usepackage{graphicx,arydshln}
\usepackage[thinlines]{easybmat}
\usepackage{amsmath}

\begin{document}
\newtheorem{assumption}{Assumption}
\newtheorem{conclusion}{Conclusion}
\newtheorem{example}{Example}
\newtheorem{remark}{Remark}

\newcommand{\rev}{} 
\newcommand{\revv}{}

\title{Diffusion of new products with recovering consumers}
\date{\today}
 \author{G. Fibich\thanks{School  of Mathematical Sciences,
   Tel Aviv University, Tel Aviv 69978, Israel (fibich@tau.ac.il, www.math.tau.ac.il/$\sim$fibich).}}

\maketitle
\slugger{siap}{siap}{xx}{x}{x--x}
\begin{abstract}
We consider the diffusion of new products in the discrete Bass-SIR model, in which consumers who adopt the product can later ``recover'' and stop influencing their peers to adopt the product. To gain insight into the effect of the social network structure  on  the diffusion, we focus on two extreme cases. In the ``most-connected'' configuration where all consumers are inter-connected (complete network), averaging over all consumers leads to an aggregate model, which combines the Bass model for diffusion of new products with the SIR model for epidemics. In the ``least-connected'' configuration where consumers are arranged on a circle and each consumer can only be influenced by his left neighbor (one-sided 1D network), averaging over all consumers leads to a different aggregate model which is linear, and can be solved explicitly. We conjecture that for any other network, the diffusion is bounded from below and from above by that on a one-sided 1D network and on a complete network, respectively. When consumers are arranged on a circle and each consumer can be influenced by his left and right neighbors (two-sided 1D network), the diffusion is strictly faster than on a one-sided 1D network. This is  different from the case of non-recovering adopters, where the diffusion on one-sided and on two-sided 1D networks is identical. We also propose a nonlinear model for recoveries, and show that consumers' heterogeneity has a negligible effect on the aggregate diffusion. 

\end{abstract}

\begin{keywords} marketing, diffusion in social networks, Bass model, SIR model, agent-based models \end{keywords}
\begin{AMS}91B99,92D25,90B60\end{AMS}

\newcommand{\IS}{\underline{\mbox{IS}}}
\newcommand{\SaS}{\underline{\mbox{SS}}}
\newcommand{\ISk}{{\underline{\mbox{I$S_k$}}}}
\newcommand{\ISkI}{{\underline{\mbox{I$S_k$I}}}}
\newcommand{\SkI}{{\underline{\mbox{$S_k$I}}}}
\newcommand{\RSk}{{\underline{\mbox{R$S_k$}}}}
\newcommand{\SSk}{{\underline{\mbox{S$S_k$}}}}
\newcommand{\ISSk}{{\underline{\mbox{IS$S_k$}}}}
\newcommand{\Sk}{{\underline{\mbox{$S_k$}}}}
\newcommand{\Skplusone}{{\underline{\mbox{$S_{k+1}$}}}}
\newcommand{\ISkplusone}{{\underline{\mbox{I$S_{k+1}$}}}}
\newcommand{\ISkplusoneI}{{\underline{\mbox{I$S_{k+1}$I}}}}
\newcommand{\SkplusoneI}{{\underline{\mbox{$S_{k+1}$I}}}}
\newcommand{\RSkplusone}{{\underline{\mbox{R$S_{k+1}$}}}}
\newcommand{\Skplustwo}{{\underline{\mbox{$S_{k+2}$}}}}
\newcommand{\SI}{\underline{\mbox{SI}}}
\newcommand{\ISI}{\underline{\mbox{ISI}}}
\newcommand{\ISS}{\underline{\mbox{ISS}}\,}
\newcommand{\ISR}{\underline{\mbox{ISR}}\,}

\newtheorem{obsrv}{Observation}
\newtheorem{conj}{Conjecture}

\newcommand{\Remark}{\vspace{0mm} \parindent=0pt
         {\bf Remark.} \hspace{0mm} \parindent=3ex}

\section{Introduction}

Diffusion of new products is a fundamental problem in marketing research. The diffusion 
begins when the product is first introduced into the market, and progresses as
consumers adopt the product. Here, to {\em adopt} the product means to buy it (e.g., Ipad), download it (e.g., Skype), try it  (e.g., Google search), use it (e.g., Facebook), etc.

The first quantitative model of diffusion of new products
was proposed in~1969 by Bass~\cite{Bass-69}. In this model, the diffusion 
depends on two parameters, $p$ and~$q$, which correspond to the likelihood of a non-adopter to adopt the product 
due to {\em external influences} by mass media or commercials, and due to {\em internal influences} by
individuals who already adopted the product (peer effects/word of mouth), respectively. 
The Bass model inspired a
huge body of theoretical and empirical research~\cite{Mahajan-93}
(in 2004 it was named one of the ten most-cited papers in the 50-year
history of Management Science~\cite{Hopp-04}).
 Most of its extensions, however, were aggregate (macroscopic) models. 
More recently, diffusion of new products has been studied 
using discrete, agent-based models (ABM)~\cite{OR-10,Garcia-05,Goldenberg-09,Goldenberg-07,GLM-01,GLM-10}. 
This kinetic-theory approach has the advantage that it reveals the relation between the (microscopic) behavior of individual consumers 
and the aggregate market diffusion,
and allows individual-level heterogeneity within both adoption decisions and
social networks~\cite{Rand-11}.

In~\cite{Bass-SIR-model-16} we introduced the discrete Bass-SIR model 
for the diffusion of new products.
Unlike previous models, it allows for the possibility that adopters stop influencing their peers after some time. This can occur because they bought the product but stopped using it, because they stopped discussing it with their friends, because their friends became indifferent to their influence, etc. 
The motivation for this model came from a recent study in which Graziano and
Gillingham empirically examined the adoption of solar photovoltaic systems in Connecticut~\cite{Graziano-15}. 
They observed a strong relationship between adoption and the number of nearby previously installed systems. 
In particular, they noted that this  effect of nearby systems {\em  diminished with
time}. This temporal decay of internal influence can be attributed to any of the above reasons.
 In addition, most people who install solar panels 
put a small sign in their front yard announcing the installation. Over time, some of these signs probably do not survive.  
Another empirical evidence for the temporal decay of internal influence follows from  
Banerjee et al.~\cite{Banerjee-13,Banerjee-14} who studied a diffusion model in which information is only passed for a finite number of iterations.  
They found that when using the model, the finite duration of passing information makes a big difference, and that by including the limits on passing information, the model much more closely matches the data.   

The possibility that adopters become non-contagious was previously considered in 
studies that used SIR-type models.  
As pointed out in~\cite{Bass-SIR-model-16}, however,  
{\em the SIR model is inappropriate for diffusion of new products, and its diffusion dynamics
is very different from that of the Bass-SIR model.}
 In particular, in
 the SIR model, there is a threshold quantity which determines whether an epidemic occurs or the disease simply dies out.
 In contrast, in the Bass-SIR model everyone ultimately adopt the product, since even in the absence of internal influences, all consumers eventually adopt due to external effects. 
This does not mean that the entire population will end up adopting the product, but rather that 
the Bass-SIR model only takes into account the people in the population that ultimately adopt the product (the ``market potential''). Thus, the Bass-SIR model is concerned with the rate at which the aggregate diffusion takes place.  
For example, a typical application of the Bass-SIR model is to compute the market half-life time~$T_{1/2}$ 
at which the product would  be adopted by~50\% of its market potential, and to determine how~$T_{1/2}$ depends on the network structure and on the recovery rate~$r$.

The focus of this paper is on analyzing the effect of the social network structure on the aggregate 
diffusion dynamics in the discrete Bass-SIR model, which is a continuous-time Markov chain (CTMC).
The paper is organized as follows.
In Section~\ref{sec:review} we review the discrete Bass-SIR model for 
diffusion of new products with recovering adopters.
In Section~\ref{sec:prem-analysis}  we obtain explicit solutions 
for the case of purely-external adoptions ($q=0$). Then, for $q>0$ 
  we show that the effect of recovery on the diffusion depends
 on the dimensionless variable~$r/q$, where~$r$ is the probability of recovery. 
Thus, when $r \ll q$, recovery only leads to a slightly slower diffusion, whereas
when $r \gg q$, diffusion is much slower, and is similar to that in the case of purely-external adoptions.

In Section~\ref{sec:1D} we consider periodic 1D networks where consumers are located on a circle.
We show that the one-sided 1D model, where each consumer can only be influenced by his left neighbor, 
reduces in the limit of an infinite number of consumers to the
{\em one-sided 1D Bass-SIR model}.  
This novel reduced model for the aggregate diffusion dynamics consists of four linear ODEs, which can be  solved 
explicitly for the aggregate adoption curve.  When each consumer can be influenced by his left and right neighbors,
taking the limit of an infinite number of consumers leads to 
a different aggregate model, the {\em two-sided 1D Bass-SIR model}.
We show analytically and numerically that
diffusion in the two-sided case is slightly faster than in the one-sided case.
This result is surprising, since in the absence of recovery, the diffusion  
is identical in both cases.

In the case of a nonspatial (complete) network where all consumers are inter-connected, averaging over all consumers leads to the {\em nonspatial Bass-SIR model}, which combines the Bass model for diffusion of new products with the SIR model for diffusion of epidemics. 
Since the one-sided 1D network and the nonspatial network are the least- and most-connected networks, respectively, we conjecture that for 
any other connected network, diffusion is faster than in the one-sided 1D~model, and slower than in the nonspatial model (Section~\ref{sec:dimensionality}). 
This conjecture is supported numerically for $D$-dimensional Cartesian networks, scale-free networks, and small-worlds networks. 
As the probability for recovery~$r$ increases, 
internal (word-of-mouth) effects become weaker. As a result,  
the half-life time~$T_{1/2}$ increases,
and the dependence of the diffusion dynamics on the network structure decreases.
Nevertheless, the dependence of~$T_{1/2}$ on the network structure {\em increases} with~$r$ for mild 
values of~$r$.

The assumption that consumers are homogeneous is convenient for the analysis. A more realistic assumption, however, is that each consumer has different parameters~$p_i$, $q_i$, and~$r_i$.
We show that in the case of vertex-transitive networks, 
the difference between diffusion in the heterogeneous and homogeneous cases is quadratically 
small in the level of heterogeneity (Section~\ref{sec:heterogeneous}). 
Indeed, our simulations reveal that even with 50\% heterogeneity,
diffusion in the heterogeneous case is only slightly lower than the homogeneous one, both for vertex-transitive networks (periodic Cartesian networks), and for 
networks which are not  vertex-transitive (small-word, scale-free).

In Section~\ref{sec:nl-recovery} we relax the assumption that adopters recover independently of other adopters.
In analogy with the Bass model, we introduce a nonlinear recovery model, in which consumers can recover both externally (i.e., independently of other consumers) and internally (i.e., because of interactions with recovered/dissatisfied consumers).  
This situation arises in online social networks, where some people leave the social network because they are unhappy with it, while others 
leave it because their friends are no longer there. 

Obviously, social networks are neither complete not one-dimensional. Nevertheless,  
the above results should be relevant to diffusion on real social networks. Indeed, in~\cite{Bass-SIR-model-16} we showed that a small-worlds structure has a negligible effect on the diffusion, and that diffusion on scale-free networks is equivalent to that on Cartesian grids. 
In addition, in the case of solar photovoltaic systems, the adoption is predominantly influenced by nearby previously installed systems~\cite{Graziano-15},  and so the social network in essentially a 2D Cartesian grid. 
The 1D Cartesian grid which is analyzed in this study, therefore, in a reasonable toy model, which has the advantage that it is amenable to analysis.  
Finally, our results for the 1D and a complete networks are conjectured to provide lower and upper bounds for the diffusion in any social network.

\section{Discrete Bass-SIR model}
\label{sec:review}

Our starting point is the discrete Bass-SIR model for diffusion of new products with recovering consumers, which was recently introduced in~\cite{Bass-SIR-model-16}. A new product is introduced to a market with $M$~consumers at time $t=0$.
Initially all consumers are non-adopters.  
If a consumer adopts the product, he becomes a contagious adopter. 
A contagious adopter can later ``recover'' and become a non-contagious adopter. 
The consumers belong to a social network which is described by an undirected or directed graph.
Let~$k_j$ denote the number of consumers connected to consumer~$j$
(the ``degree'' or  ``indegree'' of node~$j$, respectively),
and assume that there are no ``isolated'' consumers (i.e., $k_j \ge 1$ for all~$j$).
 If~$j$ did not adopt the product by time~$t$,
his probability to adopt (and thus become contagious) in~$(t,t+\Delta t)$ is
\begin{subequations}
\label{eq:discrete-Bass-SIR}
\begin{equation}
\label{LinearAdoptionRates-new}
\text{Prob}{j~\text{adopts in}\choose {(t,t+\Delta t)}}= \left(p+q \frac{i_j(t)}{k_j}\right) \Delta t+o(\Delta t) 
\qquad \mbox{as $\Delta t\to0$},
\end{equation}
 where~$i_j(t)$ is the number of contagious adopters connected to~$j$ at time~$t$. 
 The parameters~$p$ and~$q$ describe the likelihood of an individual to
adopt the product due to {\em external influences} such as mass
media or commercials, and due to {\em internal influences} by
contagious consumers who have already adopted the product ({\em word of mouth}, {\em peer effects}),
respectively.

The magnitude of internal influences experienced by~$j$ increases linearly with the number~$i_j$  
 of contagious adopters connected to~$j$, and is normalized by~$k_j$, see~\eqref{LinearAdoptionRates-new},  so that
 regardless of the network structure, the maximal internal influence that~$j$ can experience (when all his social connections are contagious adopters) is~$q$. The normalization by~$k_j$ 
allows for a meaningful comparison of the effect of the network structure. Indeed, in the absence of normalization [i.e., if we set $k_j=1$ in~\eqref{LinearAdoptionRates-new}], it is trivial that adding more connections to a network leads to a faster diffusion.
With the normalization by~$k_j$, however, it is not clear e.g., whether diffusion in the one-sided 1D case is slower than in the 
two-sided 1D case (see Section~\ref{sec:one-side>two-side}).

As in the SIR model, we assume that if~$j$ was  a contagious adopter  at time~$t$,
his probability  to recover and become non-contagious  in~$(t,t+\Delta t)$  is
\begin{equation}
\label{eq:recover-rate}
\text{Prob}{j~\text{recovers in}\choose {(t,t+\Delta t)}} = r \Delta t+o(\Delta t)
\qquad \mbox{as $\Delta t\to0$}.
\end{equation}
\end{subequations}
In Section~\ref{sec:nl-recovery} we consider a more general model for recovery.

We denote the fraction of non-adopters (``Susceptible''), contagious adopters (``Infected''), and
non-contagious adopters (``Recovered'') at time~$t$ by~$S(t)$, $I(t)$, and $R(t)$, respectively.   
The fraction of adopters (contagious and recovered)  is 
$$
f(t) = I(t) + R(t)  = 1-S(t).
$$
Since the product is new, initially all consumers are non-adopters, and so 
\begin{equation}
\label{eq:IC-SIR-general}
S(0) = 1, \qquad I(0) = R(0)  = f(0) = 0.
\end{equation}

\section{Preliminary  analysis}
\label{sec:prem-analysis}

In general, the effect of internal influences on the adoption curve~$f(t)$ 
depends on the structure of the social network.
In this section we derive some results that hold for any network.

\subsection{Purely-external adoptions}

 In the absence of internal effects ($q=0$),  relation~\eqref{LinearAdoptionRates-new} reduces to
\begin{equation}
\label{LinearAdoptionRates-p-only}
\text{Prob}{j~\text{adopts in}\choose {(t,t+\Delta t)}}= p \Delta t+o(\Delta t),
\qquad \mbox{as $\Delta t\to0$}.
\end{equation}
Therefore, the equations for $S$, $I$, and $R$ read
\begin{subequations}
\label{eq:Bass_SIR-q=0}
\begin{equation}
 S'(t) = -pS , \qquad I'(t) = pS -r I, \qquad  R'(t) = r I,
\end{equation}
where $\mbox{}' = \frac{d}{dt}$, subject to
\begin{equation}
\label{eq:Bass_SIR_ic-q=0}
S(0) = 1, \qquad I(0) = 0, \qquad R(0)  = 0.
\end{equation}
\end{subequations}

The solution of~\eqref{eq:Bass_SIR-q=0} is
\begin{subequations}
\label{eq:ext}  
\begin{equation}
S = S^{\rm ext}(t):= e^{-p t}, \quad 
I = I^{\rm ext} :=  p \frac{e^{-rt}-e^{-pt}}{p-r}, \quad 
R = R^{\rm ext} :=1-\frac{pe^{-rt}-re^{-pt}}{p-r}.
\end{equation}
In particular,
\begin{equation}
\label{eq:Sf_ext}  
 f =f^{\rm ext}(t; p):=1-e^{-pt}.
\end{equation}
\end{subequations}
Thus, in the absence of internal effects, recovery does not affect the adoption curve. Recovery does affect, however, the partition of 
adopters into contagious and recovered ones. For example,
$$
 I^{\rm ext} \sim e^{-rt}-e^{-pt}, \quad R^{\rm ext} \sim 1-e^{-rt}, \qquad r \ll p, 
$$
and
$$
 I^{\rm ext} \sim \frac{p}{r} \left(e^{-pt}-e^{-rt}\right), \quad R^{\rm ext} \sim 1-e^{-pt}, \qquad r \gg p.
$$

Once internal effects are added, recovery affects the adoption curve, since the rate of new internal adoptions depends on~$I$. Indeed, by~\eqref{LinearAdoptionRates-new}, internal effects accelerate the adoption process, i.e., 
\begin{equation}
\label{eq:f>f_ext}  
f(t;p,q,r)>f^{\rm ext}(t;p), \qquad   t>0,\quad q>0.
\end{equation}
In particular, in the Bass-SIR model~\eqref{eq:discrete-Bass-SIR}, everyone eventually adopt the product.

\subsection{Dimensionless parameter $r/q$}
\label{sec:dim-analysis}

Since the case of most interest is when the new product spreads predominantly through word-of-mouth (i.e., $p \ll q$),  
we rescale time as $t^*:=qt$. 
Hence,
$$
  f(t;p,q,r) = f(t^*;p^*,r^*), \qquad   p^* := \frac{p}{q}, \quad
  r^* := \frac{r}{q}.
$$
This shows that the aggregate effect of recovery depends on the dimensionless parameter
$
r^*= \frac{r}{q}.
$ 
Since $r^*= \frac{I r \Delta t}{I q \Delta t}$, this parameter corresponds to the rate of  loss (recovery) of contagious adopters over the rate of the creation of new ones (when most consumers are still non-adopters). 
There are two limiting cases: 
\begin{itemize}
  \item  When $r \ll q$,  adopters have sufficient time to influence their neighbors before they become non-contagious. Hence,
  the effect of recovery is small, and diffusion is similar to that in the absence of recovery, i.e., 
$f(t;p,q,r) \approx f(t;p,q,r=0)$.
  
\item When $r \gg q$, adopters have little time to influence their neighbors before they become non-contagious. Hence, internal effects effectively disappear, and diffusion is driven by purely-external adoptions. Therefore, $f(t;p,q,r) \approx f^{\rm ext}(t;p)$, see~\eqref{eq:ext}.
In particular, diffusion is considerably slower than in the absence of recovery. 
\end{itemize}

Intuitively, as $r$ increases, 
internal influences last for shorter times, and therefore:
\begin{enumerate}
  \item Diffusion becomes slower, i.e., 
\begin{subequations} 
\begin{equation}
 \label{eq:df_dr<0}
\mbox{$f(t;p,p,r)$ is monotonically decreasing in $r$}.
\end{equation}
  \item Its dependence on the network structure decreases, i.e., if $f_{\rm I}$ and $f_{\rm II}$ denote the expected fractional adoption is networks~I and~II, then
\begin{equation} 
 \label{eq:df1-f2_dr<0}
\mbox{$\left|f_{\rm I}(t;p,p,r)-f_{\rm II}(t;p,p,r)\right|$ is monotonically decreasing in $r$}.
\end{equation}
\end{subequations} 

\end{enumerate}
In particular, as~$r$ increases from~0 to~$\infty$, $f$~decreases monotonically 
from $f(t;p,q,r=0)$ to 
$f(t;p,q,r=\infty) = f(t;p,q=0,r)=f^{\rm ext}(t;p)$.

\section{1D Networks}
  \label{sec:1D}

We now consider the ``least-connected'' network, namely,  
when consumers are located on a circle such that each consumer is only connected to one or two consumers.

\subsection{One-sided 1D networks}
\label{sec:One-sided-1D}

In the one-sided 1D~network, $M$ consumers are located on a circle, 
and each consumer is only influenced by his left neighbor. Since $k_j = 1$,
 relation~\eqref{LinearAdoptionRates-new} reads  
\begin{equation}
\label{LinearAdoptionRates_1D_onesided}
\text{Prob}{j~\text{adopts in}\choose {(t,t+\Delta t)}}= \left(p+q \, i_j(t) \right)\Delta t+o(\Delta t) 
\qquad \mbox{as $\Delta t\to0$},
\end{equation}
where $i_j(t)=1$ if $j-1$ is a contagious adopter at time~$t$, and $i_j(t) = 0$ otherwise.

A priori, finding the aggregate diffusion dynamics requires writing an ordinary differential equation
for the dynamics of each of the $3^M$ possible configurations.\footnote{Each of the $M$ consumers can be 
susceptible, infected, or recovered.} As $M \to \infty$, however, this infinite system can be
reduced to a system of 4~linear ordinary differential equations:
\begin{lemma}
\label{lem:SIR-1D}
Consider the discrete Bass-SIR model~\eqref{eq:discrete-Bass-SIR} on
a one-sided 1D network. 
As $M \to \infty$, the diffusion dynamics is governed by the {\bf one-sided 1D Bass-SIR model}
\begin{subequations}
\label{eq:1D_SIR}
\begin{equation}
\label{eq:1D_SIR-a}
S'(t) = -p S -  q \IS, \qquad 
{\IS}'(t) = p e^{-pt}S+ (q e^{-pt}-p -q -r)\IS,
\end{equation}
and
\begin{equation}
\label{eq:1D_SIR-b}
I'(t) = p S+ q \IS  -rI, \qquad
R'(t) = rI,
\end{equation}
subject to 
\begin{equation}
S(0) = 1, \qquad \IS(0) = I(0) = R(0) = 0.
\end{equation}
\end{subequations}
\end{lemma}

Here, {\IS} denotes the fraction of pairs where the left consumer is infected and  
the right consumer is susceptible.\footnote{Or equivalently, for any~$j$, the probability that 
$j$~is infected and $j+1$ is susceptible.}
Thus, $\IS \not= I \cdot S$. 
The dynamics is determined by eqs.~\eqref{eq:1D_SIR-a} for~${S}$ and ${\IS}$. Once these 2~equations are solved,  
$R$ and $I$ can be recovered from eqs.~\eqref{eq:1D_SIR-b}.

\proof 
We modify the analysis in~\cite[Section~2]{OR-10}, as follows.
Let~$(S_{k})$ denote a sequence of~$k$
adjacent non-adopters, let~$(I S_{k})$ denote a sequence of a single contagious adopter and $k$ non-adopters,
 and let~$(R S_{k})$ denote a sequence of a single recovered adopter and $k$ non-adopters,
i.e.,
 \[
(S_{k})=(\underbrace{S\ldots S}_{k\text{
times}}),
 \qquad  (I S_{k})=(I \underbrace{S\ldots S}_{k\text{
times}}),
 \qquad  (R S_{k})=(R \underbrace{S\ldots S}_{k\text{
times}}),
\] 
and let $\Sk$, $\ISk$, and~$\RSk$ denote the probabilities of these configurations at time~$t$.

A configuration $(S_{k})$ cannot be
created, as the only possible transformation is $(S) \to (I)$.
A configuration $(S_{k})$ is destroyed if:
 \begin{enumerate} 
\item Any of the rightmost $k-1$ '$S$'s turns into an~'$I$',
which happens at a rate of~$p$. 
\item A configuration $(SS_{k})$ transforms into the
configuration $(SIS_{k-1})$, which happens at a rate
of~$p$. 
\item A configuration $(IS_{k})$ transforms into the configuration $(IIS_{k-1})$,
which happens at a rate of~$p+q$.
\item A configuration $(RS_{k})$ transforms into the
configuration $(RIS_{k-1})$, which happens at a rate
of~$p$. 
 \end{enumerate}
Therefore, the equation for $\Sk$ is 
$$
\Sk'(t)=-(k-1) p\Sk-p \Skplusone-(p+q)\ISk -p \RSk, \qquad k = 1,2, \dots
$$
Since 
\begin{equation}
\label{eq:S+I+R}
 \SSk+\ISk + \RSk = \Sk, 
\end{equation}
the last equation reads 
\begin{subequations}
\label{eq:1D-SI}
\begin{equation}
\label{eq:1D-S_k}
\Sk'(t)= -kp \Sk-q \ISk , \qquad k = 1,2, \dots
\end{equation}
The motivation for~\eqref{eq:1D-S_k} is as follows. Any~$S$ can change to~$I$ at the rate~$p$. Therefore, the overall rate of change due to external effects 
is~$kp \Sk$.
In addition, the leftmost~$S$ can change to~$I$ due to internal effects, if his left neighbor is an~I.  
Therefore, the overall rate of change due to external effects is~$q \ISk $.

A configuration $(IS_{k})$ is created from $(SSS_{k})$ at a rate~$p$,
from $(ISS_{k})$ at a rate~$p+q$ and from $(RSS_{k})$ at a rate~$p$.
A configuration $(IS_{k})$ is destroyed:
 \begin{enumerate} 
\item When any of the rightmost $k-1$ '$S$'s turns into an~'$I$',
which happens at a rate of~$p$. 
\item When the left~$S$ changes to~$I$ at a rate
of~$p+q$. 
\item When the $I$ changes to an~$R$ at a rate of~$r$.
 \end{enumerate}
Therefore, the equation for $I S_k$ is 
$$
\ISk'(t) = p \Skplustwo +(p+q) \ISkplusone + p \RSkplusone  -\big((k-1)p +(p+q) +r\big) \ISk , \qquad k = 1,2, \dots
$$
Therefore, by~\eqref{eq:S+I+R},
\begin{equation}
  \label{eq:ISk}
\ISk '(t) = p  \Skplusone+ q  \ISkplusone  -(kp +q +r) \ISk , \qquad k = 1,2, \dots
\end{equation}
The motivation for~\eqref{eq:ISk} is as follows. $(IS_k)$ are created from~$(SS_k)$ at a rate of~$p \SSk$ due to external effects,
 and $q  \ISSk$ due to internal effects. 
Any~$S$ can change to~$I$ at the rate~$p$. Therefore, the overall rate of change due to external effects 
is~$k p \ISk$. The leftmost~$S$ can change to~$I$ due to internal effects
at the rate of $q \ISk$. The~$I$ can change to~$R$ at the rate of~$r \ISk$.

 Since there are no adopters at $t=0$, the initial conditions are
\begin{equation}
   \label{eq:ISk-ic}
\Sk(t=0) = 1, \quad \ISk(t=0) = 0, \qquad k = 1,2, \dots
\end{equation}
\end{subequations}
Therefore, the dynamics is governed by~\eqref{eq:1D-SI}.
This infinite system can be reduced to two coupled ODEs via the 
substitution\footnote{This reduction is not possible for general initial conditions. ``Fortunately'', 
it is possible for the initial conditions~\eqref{eq:ISk-ic}, which follow from~\eqref{eq:IC-SIR-general}.} 
\begin{equation}
 \label{eq:Sk=exp(-kpt)*x} 
  \Sk = e^{-k p t} x(t), \qquad 
  \ISk = e^{-k p t} y(t), \qquad k = 1,2, \dots
\end{equation}
Indeed, the equations for $x$ and $y$ read 
\begin{subequations}
\label{eq:dot-x_ab}
\begin{equation}
\label{eq:dot-x}
  x' = -q y, \qquad y' =  p e^{-pt}x +(qe^{-pt} - q - r) y,
\end{equation}
subject to 
\begin{equation}
   \label{eq:ic-xy}
    x(0) = 1, \qquad y(0) = 0.
\end{equation}
\end{subequations}

The equation for $S'$ follows from~\eqref{eq:1D-S_k} with $k=1$. 
By~\eqref{eq:ISk} with $k=1$ and~\eqref{eq:Sk=exp(-kpt)*x},
$$
{\IS}'(t) = p \SaS+ q \ISS  -(p +q +r)\IS  = 
p e^{-pt}S+ q e^{-pt}\IS  -(p +q +r)\IS.
$$
The equation for $I'$ is {\em not} given by~\eqref{eq:ISk} with $k=0$.\footnote{This  because when $k=0$, there is no ``left~$S$ that changes to~$I$ at a rate
of~$p+q$''.} Rather, a derivation similar to that of shows that $I' = p S+ q \IS  -rI$.
Finally, since $S+I+ R = 1$, then  $R' = -S'-I'$. 
\endproof

The one-sided 1D Bass-SIR model~\eqref{eq:1D_SIR} ``identifies'' with the nonspatial Bass-SIR model~\eqref{eq:Bass_SIR} if one  makes the approximation 
$\IS \approx I \cdot S$. This mean-field approximation, however, is very inaccurate,
 especially when $q \gg p$~\cite{OR-10}. 
Indeed, the diffusion dynamics in these models can be quite different (see, e.g., 
Figs.~\ref{fig:cellular_nonspatial_1D_2D_3D_main} and~\ref{fig:T_half}B). 

Unlike the nonspatial Bass-SIR model~\eqref{eq:Bass_SIR}, the one-sided 1D Bass-SIR model~\eqref{eq:1D_SIR} is {\em linear}. In fact, 
 we can solve it explicitly: 
\begin{lemma}
\label{lem:1D-1-sided}
Consider the discrete Bass-SIR model~\eqref{eq:discrete-Bass-SIR} on
a one-sided 1D network. 
 Then $\lim_{M \to \infty}f(t) = f_{\rm 1D}^{\rm one-sided}(t)$, 
 where
\begin{subequations}
   \label{eq:f-1D-ab}
\begin{eqnarray}
   \label{eq:f-1D}
  f_{\rm 1D}^{\rm one-sided} (t) &=& 1-
e^{-(r+q+p)t +\frac{q}{p} (1-e^{-pt})} 
   \left( 1+ r  \int_0^{t} e^{(r+q) \tau -\frac{q}{p} (1-e^{-p\tau})} \, d\tau \right)
\\
 &=& 1-e^{-pt}+e^{-pt -g(t)} q \int_0^{t} e^{-r(t-\tau)} e^{g(\tau)}(1-e^{-p\tau}) \, d\tau,
\label{eq:f1D-alternate}
\end{eqnarray}
\end{subequations}
and $ g(t) = qt -\frac{q}{p} (1-e^{-pt})$.
\end{lemma}
\begin{proof} 
By~\eqref{eq:dot-x_ab},
\begin{eqnarray*}
  \ddot{x} &=& - q y' = - q \left( -(r+q) y +p e^{-pt} x + qe^{-pt} y\right)
=   -(r+q)  x' -qp e^{-pt} x + qe^{-pt}  x'
\\ &=&  -(r+q)  x' +q (e^{-pt} x)',
\end{eqnarray*}
and
$
x'(0) = -q y(0) = 0.
$
Integrating, one obtains   
$
  x'  = -(r+q) x + q e^{-pt} x +r.
$
We can rewrite this as 
$  x' -h(t) x = r$, where $h(t) = -(r+q)  + q e^{-pt}$. 
The solution of this first-order linear ODE is
$
  x(t) = e^{\int_0^t h(s) \, ds} \left( 1+ r  \int_0^t e^{-\int_0^{\tau} h(s) \, ds} \, d\tau \right),
$
where
$
e^{\int_0^t h(s) \, ds} = 
e^{\int_0^t [-(r+q)  + q e^{-ps}] \, ds} = 
e^{-(r+q)t +\frac{q}{p} (1-e^{-pt})}. 
$
And so, by~\eqref{eq:Sk=exp(-kpt)*x},
\begin{equation}
\label{eq:f_1D_1_sided_derivation}
  f_{\rm 1D}^{\rm one-sided}(t) = 1-S(t) = 1-e^{-pt} x(t),
\end{equation}
which proves~\eqref{eq:f-1D}.

 We can rewrite 
$
  f_{\rm 1D}^{\rm one-sided}(t) =  1-
e^{-rt-pt -g(t)} 
   \left( 1+ r  \int_0^{t} e^{r\tau +g(\tau)} \, d\tau \right).
$
 Now,
\begin{eqnarray*}
1+r  \int_0^{t} e^{r\tau +g(\tau)} \, d\tau  &=& 
1+  \int_0^{t} (e^{r\tau})' e^{g(\tau)} \, d\tau  = 
\\ &=& 
1+  [e^{r\tau} e^{g(\tau)}]_0^t -\int_0^{t} e^{r\tau} e^{g(\tau)}g'(\tau) \, d\tau 
= e^{r t} e^{g(t)} -\int_0^{t} e^{r\tau} e^{g(\tau)}g'(\tau) \, d\tau .
\end{eqnarray*}
 Hence,
\begin{eqnarray*}
  f_{\rm 1D}^{\rm one-sided}(t) &=&  1-
e^{-rt-pt -g(t)} 
   \left( e^{r t} e^{g(t)} -\int_0^{t} e^{r\tau} e^{g(\tau)}g'(\tau) \, d\tau  \right)
\\ &=& 1-e^{-pt}+e^{-rt-pt -g(t)} \int_0^{t} e^{r\tau} e^{g(\tau)}g'(\tau) \, d\tau, \qquad  g' = q(1-e^{-pt}),
\end{eqnarray*}
which proves~\eqref{eq:f1D-alternate}.

\end{proof}

As expected,
\begin{enumerate}
   \item When $r=0$, $f_{\rm 1D}^{\rm one-sided}=1-
e^{-(q+p)t +\frac{q}{p} (1-e^{-pt})} $, which is the expression derived in~\cite{OR-10}. 
  \item When $q=0$, $f_{\rm 1D}^{\rm one-sided} = f^{\rm ext}$, in agreement with~\eqref{eq:Sf_ext}.
 \item  $f_{\rm 1D}^{\rm one-sided}(t)$ is monotonically-decreasing in~$r$,\footnote{This follows from~\eqref{eq:f1D-alternate}, since $g(t)$ is independent of~$r$.} in agreement with~\eqref{eq:df_dr<0}.
\end{enumerate} 

\subsection{Two-sided 1D network}

In the two-sided 1D~network, $M$~consumers are located on a circle,  
and each consumer is
influenced by his left and right neighbors. 
Since $k_j = 2$, relation~\eqref{LinearAdoptionRates-new} reads  
\begin{equation}
\label{LinearAdoptionRates_1D_twosided}
\text{Prob}{j~\text{adopts in}\choose {(t,t+\Delta t)}}= \left(p+q \frac{i_j(t)}{2} \right)\Delta t+o(\Delta t) 
\qquad \mbox{as $\Delta t\to0$},
\end{equation}
where $i_j(t)=2$ if both $j-1$ and $j+1$ are contagious adopters at time~$t$,
$i_j(t)=1$ if only one of them is contagious at time~$t$,
 and $i_j(t) = 0$ otherwise. 
A priori, capturing the diffusion dynamics requires writing an ordinary differential equation
for the dynamics of each of the $3^M$ possible configurations. As $M \to \infty$, however, this infinite system can be
reduced to a system of 5~linear ordinary differential equations:

\begin{lemma}
\label{lem:SIR-1D-two_sided}
Consider the discrete Bass-SIR model~\eqref{eq:discrete-Bass-SIR} on
a two-sided 1D~network.   As $M \to \infty$, 
the diffusion dynamics is governed by the {\bf two-sided 1D Bass-SIR model}
\begin{subequations}
\label{eq:1D_SIR_two_sided}
\begin{eqnarray}
\nonumber
S'(t) &=& -p S -  q \mbox{\IS},  \\
{\IS}'(t) &=& p e^{-pt}S+ \left(\frac{q}{2} e^{-pt}-p -\frac{q}{2} -r\right)\IS-\frac{q}{2} \ISI,
\label{eq:1D_SIR_two_sided-a}
\\ 
{\ISI}'(t) &=& 2p e^{-pt}\IS+ \left(qe^{-pt}  - p-q-2r\right) \ISI,
\nonumber
\end{eqnarray}
and
\begin{equation}
\label{eq:1D_SIR-2-sided-b}
I'(t) = p S+ q \IS  -rI, \qquad
R'(t) = rI,
\end{equation}
subject to 
\begin{equation}
S(0) = 1, \qquad \IS(0)= \ISI(0) = I(0) = R(0)=  0.
\end{equation}
\end{subequations}
\end{lemma}

Here, {\ISI} denotes the fraction of triplets where the right and left consumers are infected and the center consumer
is susceptible. 
Thus, $\ISI \not= I \cdot S \cdot I$. 
The dynamics is determined by eqs.~\eqref{eq:1D_SIR_two_sided-a} 
for~${S}$, ${\IS}$, and~${\ISI}$.  Once these three equations are solved,  
$R$ and $I$ can be recovered from~\eqref{eq:1D_SIR-2-sided-b}.

\begin{proof} 
The dynamics of $\Sk$ is governed by  
\begin{subequations}
\label{eq:1D-sided}
\begin{equation}
\label{eq:1D-S_k-2-sided}
\Sk'(t)= -kp\Sk-\frac{q}{2}\ISk -\frac{q}{2} \SkI, \qquad k = 1,2, \dots
\end{equation}
Indeed, any of the~$S$ can change to~$I$ at the rate~$p$. Therefore, the overall rate of change due to external effects is $kp\Sk$. 
In addition, the leftmost~$S$ can change to~$I$ due to internal effects at the rate of~$\frac{q}{2}\ISk $, 
and the rightmost~$S$ can change to~$I$ due to internal effects at the rate of~$\frac{q}{2} \SkI$. 
Since by symmetry $\ISk  =  \SkI$, equation~\eqref{eq:1D-S_k-2-sided} is equivalent to~\eqref{eq:1D-S_k}. 

The equation for $\ISk$ is 
\begin{equation}
  \label{eq:ISk-2-sided}
\ISk '(t) = p  \Skplusone+ \frac{q}{2} \ISSk  -\left(kp +\frac{q}{2} +r\right)\ISk- \frac{q}{2} \ISkI , \qquad k = 1,2, \dots
\end{equation}
Indeed, $(IS_k)$ are created from $(SS_k)$ at a rate of~$p \SSk$ due to external effects, and 
of~$\frac{q}2 \ISSk$ due to internal effects. 
Any of the~$S$ can change to~$I$ at the rate~$p$. Therefore, the overall rate of change due to external effects 
is~$k p \ISk$. The leftmost~$S$ can also change to~$I$ due to internal effects
at the rate of~$\frac{q}{2} \ISk$.
 The~$I$ can change to~$R$ at the rate of~$r \ISk$. Finally, 
the rightmost~$S$ changes to~$I$ due to internal effects
at the rate of~$\frac{q}{2} \ISkI$.
Similar arguments show that 
\begin{equation}
\ISkI '(t) = p (\ISkplusone+ \SkplusoneI)+ \frac{q}{2}(\ISkplusoneI+ \ISkplusoneI) 
    -\left(kp +q +2r\right) \ISkI , \quad k = 1,2, \dots
\end{equation}
\end{subequations}

Under the substitution
$$
\Sk = e^{-p(k-1)t} S, \quad 
\ISk = e^{-p(k-1)t} \IS, \quad 
\ISkI = e^{-p(k-1)t} \ISI, \qquad k=1,2, \dots
$$
and using the symmetry $\ISk = \SkI$, 
the infinite system~\eqref{eq:1D-sided} reduces to the equations for 
$S'$, ${\IS}'$ and ${\ISI}'$ in~\eqref{eq:1D_SIR_two_sided}.
 Similar arguments show that 
 the equation for $I'$ reads
$$
I' = p S+ \frac{q}{2} \IS+ \frac{q}{2} \SI  -rI.
$$
Since \IS = \SI, we get the equation for $I'$ in~\eqref{eq:1D_SIR-2-sided-b}.
\end{proof}

\subsection{$f_{\rm 1D}^{\rm one-sided}(t)< f_{\rm 1D}^{\rm two-sided}(t)$}
\label{sec:one-side>two-side}

In~\cite{OR-10}, Fibich and Gibori showed that
when $r=0$, the diffusion curves in the one-sided and two-sided 1D~models are identical,
i.e., 
$$
f_{\rm 1D}^{\rm one-sided}(t;p,q)\equiv f_{\rm 1D}^{\rm two-sided}(t;p,q).
$$
Intuitively, this is because external adoptions are independent of the network structure, and 
internal adoptions occur through the expansion of 1D~clusters (chains) of adopters. 
Since the internal effect of such a chain in~$(t, t+\Delta t)$ is~$q\Delta t$ in the one-sided model and  
$\frac{q}{2}\Delta t+\frac{q}{2}\Delta t$ in the two-sided model, 
the rates of internal adoptions are identical in both cases. Hence, 
the diffusion curves are also identical.

The above argument suggests that the diffusion curves in the one-sided and two-sided 1D~models 
should remain identical when adopters are allowed to recover. Surprisingly, however,
\begin{lemma}
 \label{lem:two_sided>one_sided}
When $r>0$, diffusion in the one-sided model is strictly slower than in the two-sided model, i.e.,
$$
  f_{\rm 1D}^{\rm one-sided}(t;p,q,r)< f_{\rm 1D}^{\rm two-sided}(t;p,q,r), \qquad t>0.
$$
\end{lemma}
\begin{proof} 
 Under the substitutions~\eqref{eq:Sk=exp(-kpt)*x}  and $\ISkI = e^{- k p t} z(t)$,
 the infinite system~\eqref{eq:1D-sided} reduces to 
\begin{subequations}
\label{eq:dot-xyz}
\begin{equation}
  x' = -q y, \quad y' =  p e^{-pt}x +(\frac{q}{2}e^{-pt} - \frac{q}{2}- r) y -\frac{q}{2} z,
\quad z' =2 p e^{-pt}y+ q e^{-pt}z  - (q+2r) z,
\end{equation}
subject to 
\begin{equation}
    x(0) = 1, \qquad y(0) = 0, \qquad z(0) = 0.
\end{equation}
\end{subequations}

Since $\IS = \ISS+\ISI+\ISR$ and $\ISR>0$ for $t>0$, then
$\ISI <\IS- \ISS$, or equivalently 
$
   e^{-p t} z<e^{-p t} y-e^{-2 p t} y.
$
Therefore, the solution of~\eqref{eq:dot-xyz} satisfies 
\begin{equation}
  x' = -q y, \quad y'  >  p e^{-pt}x +(qe^{-pt} - q- r) y,   \qquad t>0,
\end{equation}
subject to~\eqref{eq:ic-xy}. 
We now follow the derivation of~\eqref{eq:f_1D_1_sided_derivation} from~\eqref{eq:dot-x_ab}, but replace the equality sign with an inequality wherever needed. Thus, we get that
the solution~\eqref{eq:dot-xyz} satisfies
$  \ddot{x} = - q y' <  -(r+q)  x' +q (e^{-pt} x)'.
$
Integrating, one obtains  
$
  x'  < -(r+q) x + q e^{-pt} x +r.
$
We can rewrite this as 
$
  x' -h(t) x< r.
$
Integrating again, one obtains
$$
  x(t)< e^{\int_0^t h(s) \, ds} \left( 1+ r  \int_0^t e^{-\int_0^{\tau} h(s) \, ds} \, d\tau \right).
$$
Hence,
\begin{eqnarray*}
  f_{\rm 1D}^{\rm two-sided}(t) &=& 1-S(t) = 1-e^{-pt} x(t)>  1-e^{-pt}   e^{\int_0^t h(s) \, ds} \left( 1+ r  \int_0^t e^{-\int_0^{\tau} h(s) \, ds} \, d\tau \right) 
 \\ &=&  f_{\rm 1D}^{\rm one-sided}(t).
\end{eqnarray*}

\end{proof}

Intuitively, once recovery occurs, the periodic 1D network is broken into several non-periodic 1D networks that do not communicate with each other. 
As we will show elsewhere, on 1D networks which are not periodic, diffusion in the two-sided case is strictly faster than in the one-sided case,
thus explaining Lemma~\ref{lem:two_sided>one_sided}.
Finally, we note that the difference between the one-sided and two-sided models is quite small (Figures~\ref{fig:cellular_1D_main}A 
and~\ref{fig:T_half}).

\subsection{Simulations}
 
\begin{figure}[ht!]
\begin{center}
\scalebox{0.7}{\includegraphics{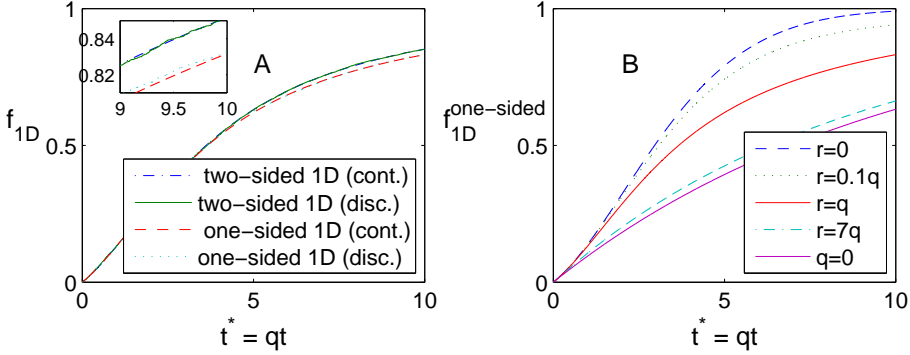}}
\caption{Fraction of adopters on 1D~networks with $p=0.01$ and  $q= 0.1$. 
A)~Agreement between [a single  simulation of] the discrete Bass-SIR model in a one-sided 1D~network [dots] 
and the continuous one-sided 1D Bass-SIR model~\eqref{eq:1D_SIR} [dashes], and between [a single
simulation of] the discrete Bass-SIR model on a two-sided 1D~network [solid]
and the continuous two-sided 1D Bass-SIR model~\eqref{eq:1D_SIR_two_sided} [dash-dots].
 Here $r=0.05$ and $M=10,000$.  
B)~The continuous one-sided 1D Bass-SIR model~\eqref{eq:1D_SIR} with various values of~$r$. 
 Here $q=0$ is~$f^{\rm ext}$, see~\eqref{eq:Sf_ext}.
}
\label{fig:cellular_1D_main}
\end{center}
\end{figure}

   Figure~\ref{fig:cellular_1D_main}A 
confirms the agreement as $M \to \infty$ between the one-sided 1D ABM  
and the one-sided 1D Bass-SIR model, see Lemma~\ref{lem:SIR-1D},
and  between the two-sided 1D ABM 
and the two-sided 1D Bass-SIR model, 
see Lemma~\ref{lem:SIR-1D-two_sided} (the agreement is clearest in the inset). 
The diffusion in the one-sided case is (slightly) slower than in the two-sided case,
in agreement with Lemma~\ref{lem:two_sided>one_sided}.
Additional numerical support that $f_{\rm 1D}^{\rm one-sided}< f_{\rm 1D}^{\rm two-sided}$ is given in
Figure~\ref{fig:T_half}.

Figure~\ref{fig:cellular_1D_main}B shows the dependence of~$f_{\rm 1D}^{\rm one-sided}(t)$ on~$r$. 
As predicted in Sections~\ref{sec:dim-analysis} and~\ref{sec:One-sided-1D},
\begin{itemize}
  \item If $r \ll q$, diffusion  is similar to that for $r=0$.
  \item $f_{\rm 1D}^{\rm one-sided}(t;r)$ is monotonically decreasing in~$r$.
  \item If $r \gg q$, diffusion is  similar to that in the absence of internal effects.
\end{itemize}

\section{Lower and upper bounds}
   \label{sec:dimensionality}

In the case of a nonspatial (complete) network where 
all $M$~consumers are connected to each other, then 
$k_j = M-1$ and $i_j(t) = M \cdot I(t)$, and so relation~\eqref{LinearAdoptionRates-new} reads  
\begin{equation}
   \label{eq:Pj-nonspatial}
\text{Prob}{j~\text{adopts in}\choose {(t,t+\Delta t)}}= \left(p+q \frac{M}{M-1}I(t)\right)\Delta t+o(\Delta t) 
\qquad \mbox{as $\Delta t\to0$}.
\end{equation}
As $M \to \infty$, the aggregate diffusion 
 is governed 
by the {\bf nonspatial Bass-SIR model}~\cite{Bass-SIR-model-16} 
\begin{subequations}
\label{eq:Bass_SIR}
\begin{equation}
\label{eq:Bass_SIR_eq}
 S'(t) = -S (p + q I), \qquad I'(t) = S (p + q I)-r I, \qquad  R'(t) = r I,
\end{equation}
\begin{equation}
\label{eq:Bass_SIR_ic}
S(0) = 1, \quad I(0) = 0, \quad R(0)  = 0.
\end{equation}
\end{subequations}
If $r=0$, then $R=0$ and $f = 1-S$, and so 
 eqs.~\eqref{eq:Bass_SIR} reduce 
to the Bass model~\cite{Bass-69}
\begin{equation}
\label{eq:f'_Bass}
   f'(t) = (1-f)(p + q f), \qquad f(0)=0.
\end{equation}
The solution of~\eqref{eq:f'_Bass} is given by the well-known Bass formula
$
f_{\rm Bass}(t) = \frac{1-e^{-(p+q)t}}{1+(q/p)e^{-(p+q)t}}.
$
 There is no explicit solution of~\eqref{eq:Bass_SIR} for $r>0$.

The 1D and nonspatial cases are the least- and most-connected networks, respectively. Therefore, it was conjectured in~\cite{OR-10} that in the absence of recoveries, for ``any'' network, the fraction of adopters is bounded by
$
f_{\rm 1D}(t;p,q)<    f(t;p,q)< f_{\rm Bass}(t;p,q).
$
Since, however, 
in the case of recovering consumers  $f_{\rm 1D}^{\rm one-sided}(t)< f_{\rm 1D}^{\rm two-sided}(t)$, see Lemma~\ref{lem:two_sided>one_sided}, 
we modify this conjecture as follows:
\begin{conj}
  \label{conj-upper_lower_bounds}
 Consider the discrete Bass-SIR model~\eqref{eq:discrete-Bass-SIR} 
on any connected network. 
As $M \to \infty$, 
the fraction of adopters is bounded by
$$
f_{\rm 1D}^{\rm one-sided}(t;p,q,r)<    f(t;p,q,r)< f_{\rm nonspatial}(t;p,q,r),
$$
where $f_{\rm nonspatial} = 1-S$ and $S$ is the solution of~\eqref{eq:Bass_SIR}.
\end{conj}

The result of Conjecture~\ref{conj-upper_lower_bounds} is not immediate, since as we add links, the weight of each link goes down (see discussion in Section~\ref{sec:review}).
The lower bound was proved in Lemma~\ref{lem:two_sided>one_sided} for the case of the two-sided 1D network. 
In Figure~\ref{fig:cellular_nonspatial_1D_2D_3D_main} we compute the diffusion numerically for periodic $D$-dimensional Cartesian networks, where each 
node is connected to its $2D$ nearest nodes and 
$\text{Prob}{j~\text{adopts in}\choose {(t,t+\Delta t)}}= \left(p+q \frac{i_j(t)}{2D}\right) \Delta t$, see~\eqref{LinearAdoptionRates-new}.   
The diffusion in the~2D and 3D~cases is indeed faster than 
in the one-sided 1D~model but slower than in the nonspatial model, in agreement with Conjecture~\ref{conj-upper_lower_bounds}. 
The differences among the four networks decrease with~$r$, in agreement with~\eqref{eq:df1-f2_dr<0}.
In~\cite{Bass-SIR-model-16} it was observed numerically that diffusion in scale-free networks in similar, if not identical, 
to that on Cartesian grids, and that a small-worlds structure has a negligible effect on the diffusion. This suggests, therefore, that Conjecture~\ref{conj-upper_lower_bounds} holds for scale-free and small-worlds networks.

\begin{figure}[ht!]
\begin{center}
\scalebox{0.5}{\includegraphics{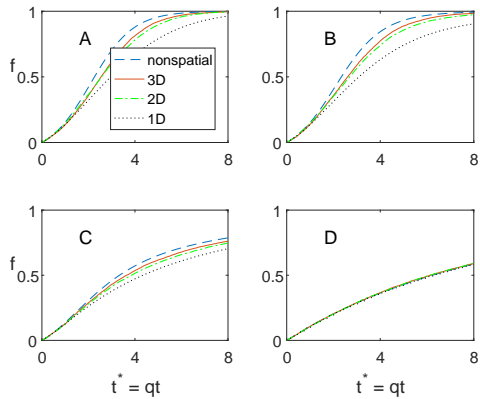}}
\caption{Fractional adoption  in one-sided 1D (dots) ,2D (dash-dot), 3D (solid), and nonspatial (dashes) networks. 
Here $p=0.01$, $q= 0.1$, and~$M=10,000$. A)~$r=0$. B)~$r=0.01$. C)~$r=0.1$. D)~$r=0.7$. 
Figure taken from~\cite{Bass-SIR-model-16}.}
\label{fig:cellular_nonspatial_1D_2D_3D_main}
\end{center}
\end{figure}

   A useful measure for comparing the diffusion in different networks is the {\em market half-life time}
 $T_{1/2} := f^{-1}(1/2)$, 
i.e., the time for half of the population to adopt. In the absence of internal effects we have that $f = f^{\rm ext}$, see~\eqref{eq:Sf_ext}, and so 
$
T_{1/2}^{\rm ext} = \frac{\ln 2}{p}.
$
By Conjecture~\ref{conj-upper_lower_bounds}, for any network with~$p$, $q$, and~$r$, 
\begin{equation}
\label{eq:ul_bounds_T1/2}
 T_{1/2}^{\rm one-sided~1D} >     T_{1/2} >  T_{1/2}^{\rm nonspatial}.
\end{equation}
Figure~\ref{fig:T_half}A shows that~\eqref{eq:ul_bounds_T1/2} indeed holds for the 
 two-sided 1D, 2D, and 3D Cartesian networks.
In addition, for all networks:
\begin{enumerate}
  \item $T_{1/2}$ is monotonically increasing in~$r$, in agreement with~\eqref{eq:df_dr<0}.
    \item 
$
T_{1/2} \to T_{1/2}^{\rm ext}$ as ${r}/{q} \to \infty$,
 since internal effects disappear in the limit (see Section~\ref{sec:dim-analysis}).
\end{enumerate}

In Figure~\ref{fig:T_half}B we plot the ratio of the upper and lower bounds in~\eqref{eq:ul_bounds_T1/2}. Surprisingly, this ratio initially increases with~$r$, and only later decreases monotonically to zero as $r/q \to \infty$. In particular,
\begin{obsrv}
   \label{observe}
When~$r$ is of a comparable magnitude to~$q$, recovery increases the dependence of~$T_{1/2}$ on the network structure.
\end{obsrv}

This observation also follows from Figure~\ref{fig:T_half}C, where we plot the ratio of the half-life times 
for the  one-sided and two-sided 1D~models. In that case, however, the maximal difference between 
the two models is~1.5\%.~\footnote{Observation~\ref{observe} may seem to contradict with Figure~\ref{fig:cellular_nonspatial_1D_2D_3D_main} that shows that
the differences among the four models decrease with~$r$.  
A closer inspection of Figure~\ref{fig:cellular_nonspatial_1D_2D_3D_main} shows that the vertical distances
between the four curves (i.e., the differences in~$f$ for a given~$t$) 
indeed decrease monotonically in~$r$. The horizontal distances
between the four curves (i.e., the differences in~$t$ for a given~$f$), however,  
initially increase with~$r$, because the curves become less steep. 
 }

\begin{figure}[ht!]
\begin{center}
\scalebox{0.7}{\includegraphics{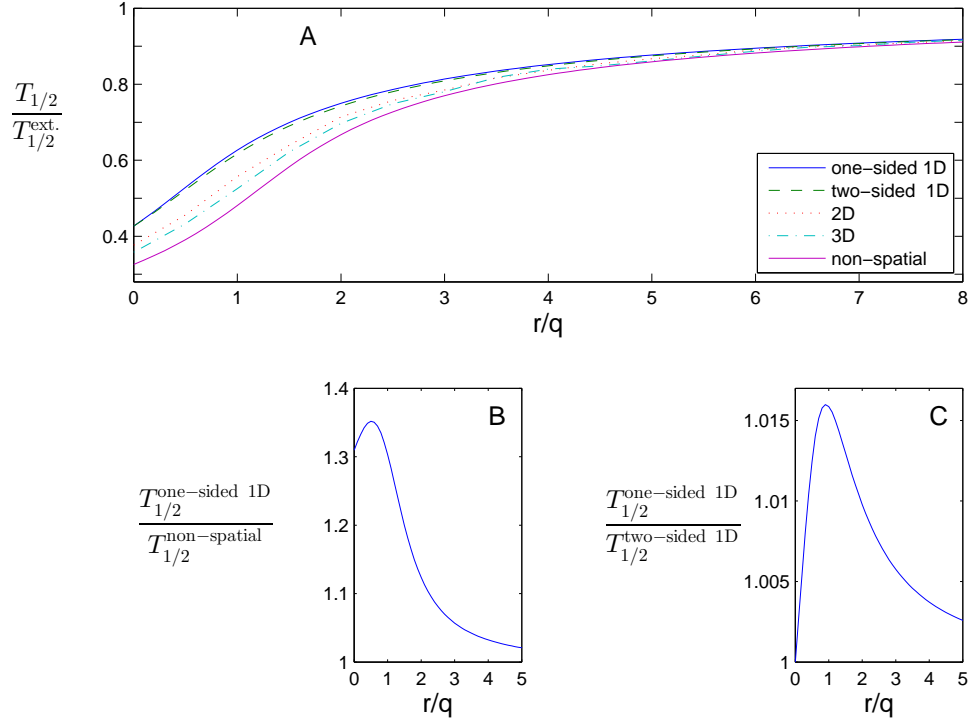}}
\caption{A)~The market lalf-time $T_{1/2}$, normalized by~$T_{1/2}^{\rm ext}$, as a function of~$r/q$, in the one-sided 1D~model (solid), two-sided 1D~model (dashes), 2D~model (dots), 3D~model (dash-dots), and the nonspatial model (solid).
Here $p = 0.01$ and $q = 0.1$. B)~The ratio $\frac{T_{1/2}^{\rm one-sided~1D}}{T_{1/2}^{\rm nonspatial}}$. 
C)~The ratio $\frac{T_{1/2}^{\rm one-sided~1D}}{T_{1/2}^{\rm two-sided~1D}}$.  }
\label{fig:T_half}
\end{center}
\end{figure}

\section{Heterogeneous consumers}
\label{sec:heterogeneous}

   So far we assumed that consumers are homogeneous, namely, they have the same~$p$, $q$, and~$r$. 
While this assumption is convenient for the analysis, a more realistic assumption is that 
consumer~$j$ has its own~$p_j$, $q_j$, and~$r_j$, i.e.,
\begin{subequations}
\label{eq:Bass-SIR-hetero}
\begin{equation}
\text{Prob}{j~\text{adopts in}\choose {(t,t+\Delta t)}}= \left(p_j+q_j \frac{i_j(t)}{k_j}\right) \Delta t+o(\Delta t) 
\qquad \mbox{as $\Delta t\to0$},
\end{equation}
and
\begin{equation}
\text{Prob}\{j\text{ recovers in }(t,t+\Delta t)\} = r_j \Delta t+o(\Delta t)
\qquad \mbox{as $\Delta t\to0$}.
\end{equation}
\end{subequations}

 Simulations with heterogeneous non-recovering consumers
on nonspatial networks and on periodic 1D and~2D Cartesian networks showed that 
heterogeneity has a minor effect~\cite{GLM-01,OR-10}. By this, we mean that the 
diffusion in the heterogeneous case was close to that in the homogeneous case with 
$\bar{p} =  \frac1M\sum_{j=1}^n p_j$ and $\bar{q} =  \frac1M\sum_{j=1}^n q_j$, even when the level of heterogeneity 
was significant. This small effect of heterogeneity  was explained in~\cite{PNAS-12} to be a consequence of 
the {\em averaging principle for heterogeneous  models}.
 Exactly the same arguments imply that 
heterogeneity  has a small effect when consumers are allowed to recover:
\begin{lemma}
Consider the heterogeneous Bass-SIR model~\eqref{eq:Bass-SIR-hetero} on
a vertex-transitive network~\footnote{A graph $G$ is called {\em vertex-transitive} if the ``view'' from any vertex is identical, 
i.e., if for given any two vertices $v_1$ and $v_2$ of~$G$, there is some automorphism $f:V(G) \rightarrow  V(G)$ such that $f(v_1) = v_2$. For example, a nonspatial network and periodic $d$-dimensional Cartesian networks are vertex-transitive.}.
Then the adoption curve satisfies
\label{lem:hetero}
$$
  f(t;p_1, \dots, p_M,q_1, \dots, q_M,r_1, \dots, r_M)  =  f(t;\bar{p},\bar{q},\bar{r})\Big(1+ O(\sigma_p^2,\sigma_q^2,\sigma_r^2)\Big),
$$
where~$\{\bar{p}, \bar{q}, \bar{r}\}$ and~$\{\sigma_p, \sigma_q, \sigma_r\}$ are the mean and standard deviation (``level of heterogeneity'') 
of~$\{p_j\}_{j=1}^M$, $\{q_j\}_{j=1}^M$, and~$\{r_j\}_{j=1}^M$, respectively.
\end{lemma}
\begin{proof} 
Following~\cite{PNAS-12}, the adoption curve $f(t;p_1, \dots, p_M,q_1, \dots, q_M,r_1, \dots, r_M)$ satisfies the following two conditions:
\begin{enumerate}
  \item $f$ is twice continuously-differentiable in $\{p_i,q_i,r_i\}_{i=1}^M$. 
  \item $f$ is weakly-symmetric in~$p$, i.e., for any $\{p, \tilde{p}, q, r\}$ and $i_0 \in \{1, \dots, M \}$,
    if $p_i = p$ for $i \not= i_0$, $p_{i_0} = \tilde{p}$, $q_i = q$ for all~$i$,
and $r_i = r$ for all~$i$, then~$f$ is independent of~$i_0$. Similarly, $f$ is weakly-symmetric in~$q$ and in~$r$.
\end{enumerate}
Indeed, condition~1 can be proved as in~\cite{PNAS-12}. Condition~2 follows from the vertex-transitive property.
Hence, the result follows from the averaging principle.
\end{proof}

In Figures~\ref{fig:cellular_1D_hetero_main} and~\ref{fig:cellular_2D_hetero_main} we present ABM simulations of the heterogeneous discrete Bass-SIR model on a periodic one-sided 1D network and on a periodic 2D~network, respectively, with 
$p_i = p(1+ \eta U(i))$  where $U$ is uniformly distributed in $[-1, 1]$,
and similarly for~$q_i$ and~$r_i$.
At the  heterogeneity level $\eta = 25\%$, the fractional adoption is nearly identical to the 
homogeneous one. Even at the heterogeneity level  $\eta = 50\%$, the aggregate adoption level is only slightly lower than in the homogeneous case~\footnote{The fact that heterogeneity {\em slows down} the diffusion can be easily proved when $q=0$. Indeed,
by~\eqref{eq:Sf_ext},  
$
f^{\rm ext}(t;p_1, \dots, p_M,r_1, \dots, r_M)  =  \frac{1}{M} \sum_{j=1}^M (1-e^{- p_j t})<
1-e^{- \bar{p} t} = f^{\rm ext}(t;\bar{p},\bar{r}), 
$
where the inequality follows form the fact that for $g(x) = 1-e^{-x}$, 
 $g''<0$, and so $ \frac{1}{M} \sum_{j=1}^M g(p_j)< g(\sum_{j=1}^M p_j)$.
 }.

\begin{figure}[ht!]
\begin{center}
\scalebox{0.5}{\includegraphics{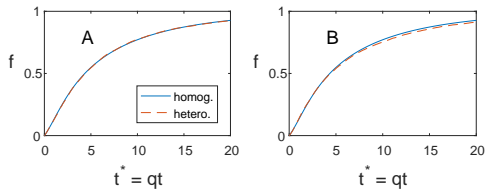}}
\caption{Fraction of adopters in [a single simulation of] the heterogeneous discrete Bass-SIR model~\eqref{eq:Bass-SIR-hetero} on a one-sided 1D~network (dashes). Dotted line
 is the homogeneous case.
Here $p=0.01$, 
$q= 0.1$, $r= 0.1$, and $M=10,000$.  Level of heterogeneity is:  A)~$\eta= 25\%$.
B)~$\eta = 50\%$. 
 }
\label{fig:cellular_1D_hetero_main}
\end{center}
\end{figure}

\begin{figure}[ht!]
\begin{center}
 \scalebox{0.5}{\includegraphics{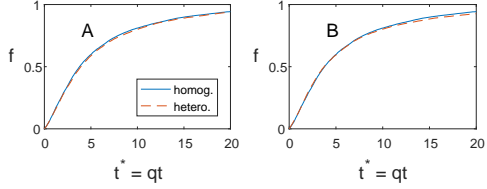}}
\caption{Same as Figure~\ref{fig:cellular_1D_hetero_main} on a 2D~network.
 }
\label{fig:cellular_2D_hetero_main}
\end{center}
\end{figure}

In Figures~\ref{fig:cellular_small_worlds_hetero_main} and~\ref{fig:cellular_scale_free_hetero_main} we present ABM simulations of the heterogeneous discrete Bass-SIR model on two networks which are not vertex-transitive: A small-world network~\cite{Watts-98}, constructed by
adding 5\% random long-range connections to a one-sided 1D network, and
 a scale-free network,  constructed using the Barab\'asi-Albert (BA) preferential-attachment algorithm~\cite{BA-99} in which each new node makes a single new link, respectively. In both cases, we again observe that heterogeneity has a negligible case on the aggregate diffusion. The result for a small-world network could be expected, since a small-world structure has a negligible effect on the diffusion in the 
Bass and Bass-SIR models~\cite{OR-10,Bass-SIR-model-16}. The result for a scale-free network in less expected, and may has to do with the 
surprising equivalence between diffusion in scale-free and Cartesian networks~\cite{Bass-SIR-model-16}. Alternatively, it may be an indication that heterogeneity has a negligible effect whenever the number of consumers is sufficiently large.

\begin{figure}[ht!]
\begin{center}
\scalebox{0.5}{\includegraphics{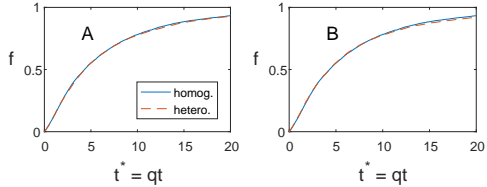}}
\caption{Same as Figure~\ref{fig:cellular_1D_hetero_main} on a small-world network.
 }
\label{fig:cellular_small_worlds_hetero_main}
\end{center}
\end{figure}

\begin{figure}[ht!]
\begin{center}
\scalebox{0.5}{\includegraphics{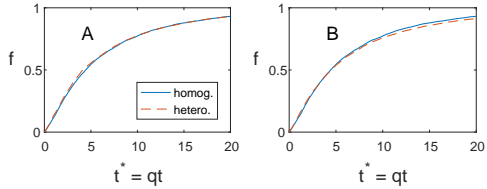}}
\caption{Same as Figure~\ref{fig:cellular_1D_hetero_main} on a scale-free network.
 }
\label{fig:cellular_scale_free_hetero_main}
\end{center}
\end{figure}

\section{Internal (nonlinear) recoveries}
\label{sec:nl-recovery}
 
In~\cite{Cannarella-14}, Cannarella and Spechler analyzed diffusion of 
online social networks, such as MySpace and Facebook.
They argued that recoveries (i.e., people leaving the social network) 
result from interactions between infected  (current) members and recovered (past) members.
Therefore, they introduced a modified SIR model on a complete network in which the relation $R' = r I$
was replaced with~\footnote{See~\cite{Bass-SIR-model-16} for why SIR models are inappropriate for diffusion of new products.} 
\begin{equation}
\label{eq:R'-Cannarella}
R' = r_{\rm nl} I R.
\end{equation}
  Since $R(0) = 0$, however, under relation~\eqref{eq:R'-Cannarella} 
 there will be no recoveries. Hence, they artificially set 
$R(0) = R_0$, where $0<R_0 \ll 1$ was a fitted parameter. To avoid this artificial fix and yet allow for nonlinear recoveries, we set
$$
R'(t) =   (r +r_{\rm nl}  R)I.
$$
 Thus, in the spirit of the Bass model, adopters can recover  independently of others (``external recoveries''), as well as  
 through interactions with recovered people (``internal recoveries''). 
 This leads to the {\bf modified Bass-SIR model}
\begin{subequations}
\label{eq:Bass_SIR-modified}
\begin{equation}
 S'(t) = -S (p + q I), \quad I'(t) = S (p + q I)-(r +r_{\rm nl}  R)I, \quad  R'(t) =   (r +r_{\rm nl}  R)I,
\end{equation}
\begin{equation}
S(0) = 1, \quad I(0) = 0, \quad R(0)  = 0.
\end{equation}
\end{subequations}

 Since $R(t) \le 1$, nonlinear internal recoveries can have a dominant effect over linear external ones 
(i.e., $r_{\rm nl}R  \gg r$), only if $r_{\rm nl} \gg r$.
To see the dynamics in this case, we 
set $r = 0.001$, so that external recoveries would have a negligible effect, and
 $r_{\rm nl} = 0.04$, so that $r_{\rm nl} \gg r$.
Since $r_{\rm nl} \ll 1$,  nonlinear recoveries become important only 
once most of the population adopts. Hence, 
the overall adoption $f = I+R$ is unaffected by the nonlinear recoveries, see Figure~\ref{fig:cellular_nonspatial_SIR_modified_dynamics}A.
Nonlinear recoveries, however, accelerate the transition from infected to recovered, 
changing it from a linear rate to an exponential one,  see Figure~\ref{fig:cellular_nonspatial_SIR_modified_dynamics}B and~\ref{fig:cellular_nonspatial_SIR_modified_dynamics}C.
Therefore, nonlinear recoveries are important if the firm only cares about the number of infected consumers (for example, if being recovered  means to stop using the product). If, however, recovered adopters bought the product
or still use it, but simply stopped promoting it, the effect of nonlinear recoveries is of much less importance to the firm.

\begin{figure}[ht!]
\begin{center}
\scalebox{0.7}{\includegraphics{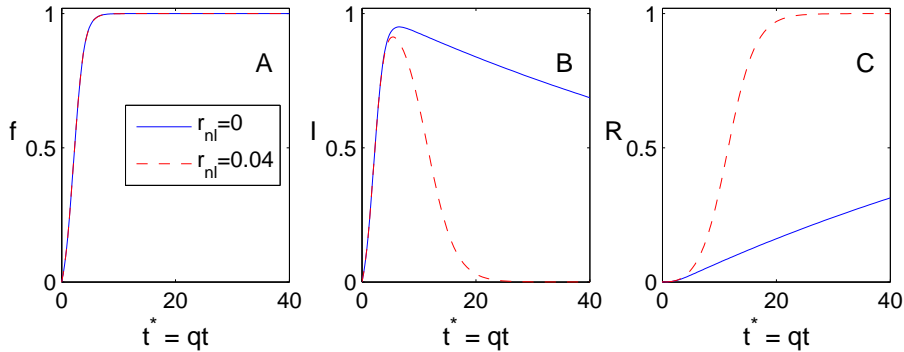}}
\caption{The modified Bass-SIR model~\eqref{eq:Bass_SIR-modified}
 with $p=0.01$,  $q= 0.1$, $r = 0.001$, and~$r_{\rm nl}=0$ (solid) or~$r_{\rm nl}=0.04$ (dashes).  
  A)~$f(qt)$. The two curves are indistinguishable. B)~$I(qt)$. C)~$R(qt)$.}
\label{fig:cellular_nonspatial_SIR_modified_dynamics}
\end{center}
\end{figure}

\subsubsection*{Acknowledgment}

We thank O. Raz, G. Ariel, and K. Gillingham for useful discussions. Part of this research was conducted 
while the author was visiting the Center for Scientific Computation and Mathematical Modeling (CSCAMM) at the University of Maryland.
This material is based upon work supported by the Kinetic Research Network (KI-Net) under NSF Grant No. RNMS  \#1107444,
and by the Department of Energy under Award Numbers DE-EE0007657.

\end{document}